\documentclass[a4paper,10pt]{article}

\usepackage[pages=all, color=black, position={current page.south}, placement=bottom, scale=1, opacity=1, vshift=5mm]{background}

\usepackage[margin=0.92in]{geometry} 

\SetBgContents{
} 
\title{Good quantum codes with addressable and parallelizable transversal non-Clifford gates}
\author{Virgile Gu\'emard$^{1,2}$}

\date{
	$^1$Aix Marseille Université, I2M, UMR 7373, 13453 Marseille, France\\%
        $^2$Inria Paris, France \\[2ex]%
	\today
}
 
\usepackage{amsmath}
\usepackage{amsthm}
\usepackage{amssymb}
\usepackage{soul}

\usepackage[utf8]{inputenc}
\usepackage{hyperref}
\hypersetup{
	unicode,
	pdfauthor={Virgile Guemard},
	pdftitle={Multi-control-Z gates},
	pdfsubject={Control-Z gates},
	pdfproducer={LaTeX},
	pdfcreator={pdflatex}
}

\usepackage[sort&compress,numbers,square]{natbib}

\theoremstyle{plain}
\newtheorem{theorem}{Theorem}[section]
\newtheorem{corollary}[theorem]{Corollary}
\newtheorem{lemma}[theorem]{Lemma}

\newtheorem{proposition}[theorem]{Proposition}

\theoremstyle{definition}
\newtheorem{definition}[theorem]{Definition}

\newtheorem{remark}[theorem]{Remark}

\usepackage{graphicx, color}
\graphicspath{{fig/}}

\usepackage{algorithm, algpseudocode} 
\usepackage{mathrsfs}

\usepackage{dcolumn}
\usepackage{bm}
\usepackage[english]{babel}
\usepackage[OT1]{fontenc}

\usepackage[colorinlistoftodos, color=green!40, prependcaption]{todonotes}

\usepackage{amsthm}
\usepackage{mathtools}
\usepackage{physics}
\usepackage{xcolor}
\usepackage{graphicx}
\usepackage{tikz-cd}
\usepackage{bbold}

\usepackage{adjustbox}
\usepackage{placeins}

\usepackage{csquotes}
\usepackage{booktabs}

\usepackage[caption=false]{subfig}
 \makeatletter
\renewcommand{\fnum@figure}{FIG. \thefigure}
\makeatother
\usepackage{tikz}
\usetikzlibrary{decorations.markings}

\usepackage{amsmath,mleftright}
\usepackage{xparse}

\NewDocumentCommand{\evalat}{sO{\big}mm}{%
  \IfBooleanTF{#1}
   {\mleft. #3 \mright|_{#4}}
   {#3#2|_{#4}}%
}

\usepackage{amsthm}

\usepackage{extpfeil}
\usepackage{tikz-cd}
\usepackage{svg}
\usepackage{nicematrix}
\usepackage{comment}

\usepackage{titlesec}
\titleformat{\section}
  {\normalfont\bfseries}{\thesection}{1em}{}
\titleformat{\subsection}
  {\normalfont\bf}{\thesubsection}{1em}{}{\normalfont\bfseries}
\titleformat{\subsubsection}
  {\normalfont\bfseries}{\thesubsubsection}{1em}{}
\renewcommand{\thesubsection}{\thesection.\arabic{subsection}}
\renewcommand{\thesection}{\arabic{section}}
\renewcommand{\thesubsubsection}{\thesubsection.\arabic{subsubsection}} 

\usetikzlibrary{arrows}

\usepackage{leftindex}
\makeatletter
\newtheorem*{rep@theorem}{\rep@title}
\newcommand{\newreptheorem}[2]{%
\newenvironment{rep#1}[1]{%
 \def\rep@title{#2 \ref{##1}}%
 \begin{rep@theorem}}%
 {\end{rep@theorem}}}
\makeatother
\newreptheorem{proposition}{Proposition}

\begin{document}
	\maketitle
\begin{abstract}
In this work, we prove that for any $m>1$, there exists a family of good qudit quantum codes supporting transversal logical $\mathsf{C}^{m-1}\mathsf{Z}$ gates that can address specified logical qudits and be largely executed in parallel. Building on the family of good quantum error-correcting codes presented in \cite{he2025good}, which support addressable and transversal logical $\mathsf{CCZ}$ gates, we extend their framework and show how to perform large sets of gates in parallel. The construction relies on the classical algebraic geometry codes of Stichtenoth \cite{Stichtenoth2006}. Our results lead to a substantial reduction in the depth overhead of multi-control-$Z$ circuits. In particular, we show that the minimal depth of any logical $\mathsf{C}^{m-1}\mathsf{Z}$ circuit involving qudits from $m$ distinct code blocks is upper bounded by $O(k^{m-1})$, where $k$ is the code dimension. While this overhead is optimal for dense $\mathsf{C}^{m-1}\mathsf{Z}$ circuits, for sparse circuits we discuss how the depth overhead can be significantly reduced by exploiting the structure of the quantum code.
\end{abstract}

\tableofcontents

\section{Introduction}

Quantum computers may allow us to efficiently solve problems for which no efficient classical algorithms are known~\cite{shor97algo,Harrow2009,Grover1996}. However, performing such computations on real devices remains extremely challenging due to the effect of noise, which corrupts information along the way. Fortunately, quantum information can be protected by quantum error-correcting codes~\cite{Calderbank1996,Shor1995,Steane1996}, and computation can be carried out using fault-tolerant schemes~\cite{Aharonov97,Gottesman2013,8555154}. The goal of fault-tolerance is to perform computations on the logical information while keeping it encoded. In particular, it has been shown that, given any quantum circuit operating with perfect qubits and gates, there exists another circuit, called a fault-tolerant circuit, that performs the same computation with high probability, provided that the noise level remains below a constant threshold. These results, known as threshold theorems, come with a resource overheads and an increase in the length of the computation.\par

Central to fault-tolerant quantum computation are \emph{logical gates} that can be implemented by low-depth physical circuits. Among them, a gate is said to act transversally on a given quantum code if it can be performed on all the logical qubits at once by acting on all the physical qubits with few-qubit gates. Transversal gates are inherently fault-tolerant, since the low connectivity of the physical circuit prevents errors from propagating. Although no quantum code can support a universal set of transversal gates~\cite{PhysRevLett.102.110502}, considerable effort has been put to finding codes that admit transversal gates up to the third level of the Clifford hierarchy~\cite{PhysRevA.95.012329}, such as $\mathsf{CCZ}$ and $T$ gates. Such constructions directly reduce the overhead of fault-tolerant schemes when combined with other techniques. \par

Among quantum codes, low-density parity-check (LDPC) codes, defined by pairs of sparse parity-check matrices, are leading candidates for physical implementations. The search for transversal gates on LDPC codes has been extensively studied, starting with Clifford gates on color codes~\cite{PhysRevLett.97.180501} and followed by non-Clifford gates~\cite{Bombin_2015,PhysRevA.91.032330,PhysRevLett.111.090505,Kubica_2015,PhysRevA.100.012312}. Recent efforts also show attempts to adapt asymptotically good quantum LDPC constructions~\cite{Panteleev2020,Panteleev2021,Leverrier2022,leverrier2022decoding,dinur2023} to devise codes supporting transversal non-Clifford gates~\cite{zhu2024nonclifford,scruby2024quantumrainbowcodes,zhu2025topological,breuckmann2024cupsgatesicohomology,golowich2024quantumldpccodestransversal,lin2024transversalnoncliffordgatesquantum}.\par

Nevertheless, it is crucial to explore the full landscape of quantum codes, including those that are not LDPC, to understand which structural properties can be leveraged to efficiently perform computations at the logical level. In this direction, non-LDPC quantum codes have been employed in \emph{magic state distillation} (MSD) protocols~\cite{PhysRevA.71.022316}, which, when combined with gate teleportation, enable universal fault-tolerant computation. The overhead of MSD has been greatly reduced since its original proposal. It was shown in~\cite{PhysRevX.2.041021,PhysRevA.86.052329,Krishna2019} that quantum codes built by puncturing triorthogonal matrices, which may be realized as generator matrices of classical Reed-Muller and Reed-Solomon codes, yield quantum codes over prime-dimensional qudits suitable for MSD protocols. More recently, these techniques were generalized using algebraic geometry (AG) codes, producing asymptotically good binary quantum codes supporting non-Clifford transversal gates~\cite{Wills2025,golowich2024asymptoticallygoodquantumcodes,nguyen2024goodbinaryquantumcodes}.\par

Furthermore, by relaxing the LDPC constraint, a recent line of work has investigated quantum codes that inherently support transversal non-Clifford logical gates which are also \emph{addressable}. This means that, given any subset of logical qudits, there exists a constant-depth physical circuit whose induced logical gate acts only on this subset. This direction was initiated in~\cite{he2025addressable}, where the authors generalized the triorthogonal matrix approach to higher orders of orthogonality, which the authors called \emph{addressable orthogonality}. This framework allows for the construction of codes supporting various types of addressable and transversal diagonal gates. In particular, using Reed-Solomon codes, they designed an explicit family of asymptotically good qudit codes supporting addressable $\mathsf{CCZ}$ gates. These qudit codes can then be converted into qubit codes at the cost of a polylogarithmic loss in the parameters. In \cite[Section 3.3]{he2025addressable} the authors also discuss conditions in which logical gates may be parallelized for quantum codes built from Reed-Solomon generating matrices.\par

Soon after, in~\cite{he2025good}, the same authors extended their previous asymptotic construction by using classical AG codes~\cite{Stichtenoth2006} together with the puncturing framework of~\cite{PhysRevX.2.041021,Krishna2019}. They further observed that a particular property of Stichtenoth’s AG codes \cite{Stichtenoth2006}, namely the existence of a transitive group action on the code coordinates, enables the implementation of $\mathsf{CCZ}$ gates acting on arbitrary triples of logical qubits. In this way, the authors of~\cite{he2025good} demonstrated the existence of a family of asymptotically good quantum codes supporting fully addressable $\mathsf{CCZ}$ gates, a major advance in the direction of fault-tolerant quantum computing.\\

In this article, we build on the work of~\cite{he2025addressable}, and in particular on the asymptotically good construction~\cite{he2025good}, and we show the existence of very large sets of addressable non-Clifford gates which can be parallelized.
To state our result precisely, let us introduce its context of application. Let $m$ be a positive integer and $q$ be a power of a prime $p$. The action of a \textit{physical $(m-1)$-control-$Z$ gate} on $m$ physical qudits in the computational basis is given as follows:
\[
\mathsf{C}^{m-1}\mathsf{Z}\ket{x_0}\ket{x_1}\dots\ket{x_{m-1}}
\coloneqq
\exp\left(\frac{2i\pi}{p} \, \mathrm{tr}(x_0x_1\dots x_{m-1})\right)
\ket{x_0}\ket{x_1}\dots\ket{x_{m-1}},
\]
where $\mathrm{tr}\coloneqq\mathrm{tr}_{\mathbb F_q/\mathbb F_p}$ is the trace map from $\mathbb F_q$ to the prime field $\mathbb F_p$ and for all $i$, $x_i\in \mathbb F_q$.
Given a quantum code $\mathcal{Q}$ over $q$-dimensional qudits, we may consider the code $\mathcal{Q}^{\otimes m}$, composed of $m$ copies of $\mathcal{Q}$. We refer to each copy of the original code as a \emph{code block} of $\mathcal{Q}^{\otimes m}$. Moreover, we call a \textit{logical interblock $(m-1)$-control-$Z$ gate} a circuit composed exclusively of logical $m$-qudit gates of the form $\overline{\mathsf{C}^{m-1}\mathsf{Z}}$, each of which involves at most one logical qudit in each code block. In this work, we obtain the following theorem.

\begin{theorem}\label{Theorem, main}
Let $\ell$ be a prime power, and set $q = \ell^2$. Then, for any integer $m\geq 2$ satisfying $\ell \geq 2(m+1)$, there exists an asymptotically good family of CSS codes $\mathcal{Q}$ with parameters $[[n,k,d]]_q$ over $q$-dimensional qudits such that, for every integer $\widetilde m$ with $1 < \widetilde m \le m$, the depth of any logical circuit on $\mathcal{Q}^{\otimes \widetilde m}$ composed solely of inter-block $(\widetilde m-1)$-control-$Z$ gates is upper bounded by $O(k^{\widetilde m-1})$.
\end{theorem}

By considering the case where $q$ is even, we can also use the method of \cite{nguyen2024goodbinaryquantumcodes}, which was adapted from \cite{Cramer09,Cramer2018}, to convert the good qudit codes into good qubit codes, at the cost of worsening the constant factors of the asymptotic scaling.\\

This article is organized as follows. In Section \ref{preliminaries}, we review important structural properties of classical codes, as well as introductory notions of function fields. We also describe the asymptotically optimal transitive AG code family of Stichtenoth \cite{Stichtenoth2006}.\par
In Section \ref{Good quantum codes with addressable and parallelizable}, we present the quantum code construction, following closely the method and presentation of \cite{he2025good}. We also show in this section that various sets of logical multi-control-$Z$ gates can be addressed in parallel, thus proving Theorem \ref{Theorem, main}.

\section*{Acknowledgement}

We would like to thank Adam Wills for pointing out Section 3.3 of his work \cite{he2025addressable}, as well as Michael Vasmer and Anthony Leverrier for valuable discussions throughout this work. We acknowledge the Plan France 2030 through the project NISQ2LSQ ANR-22-PETQ-0006.

\section{Preliminaries}\label{preliminaries}

\subsection{Linear codes with transitivity and multiplication properties}

In this work, we consider linear codes over the finite field $\mathbb{F}_q$, where $q$ is a square. Recall that a linear $[N,K,D]_q$ code $C$ is defined as a $K$-dimensional subspace $C \subseteq \mathbb{F}_q^N$. The distance $D$ is the minimal Hamming distance between any two vectors of $C$. To obtain our result, we will need codes with additional structure, which we now review.

\paragraph{Transitive codes.}
Given a permutation $\sigma \in S_N$ and a vector $x = (x_1, \ldots, x_N) \in \mathbb{F}_q^N$, we denote by $\sigma(x)$ the vector $(x_{\sigma(1)}, \ldots, x_{\sigma(N)})$. A permutation $\sigma \in S_N$ is an automorphism of the code $C \subseteq \mathbb{F}_q^{N}$ if, for all $c \in C$, we have $\sigma(c) \in C$. The set of all automorphisms of the code $C$ forms a group called the \emph{automorphism group}, denoted ${\rm Aut}(C) \subseteq S_N$.
In this work, we require codes whose automorphism group fulfills the following condition.

\begin{definition}
We call a subgroup $G$ of the symmetric group $S_N$ \emph{transitive} if, for any pair $(i,j)$ with $i,j \in \{1, \ldots, N\}$, there exists a permutation $\sigma \in G$ such that $\sigma(i) = j$. A code $C$ over $\mathbb{F}_q$ of length $N$ is said to be \emph{transitive} if its automorphism group ${\rm Aut}(C)$ is a transitive subgroup of $S_N$.
\end{definition}

\paragraph{Codes with multiplication properties.}
Of particular interest in this work are codes satisfying a strong condition of self-orthogonality. Given any two vectors $x$ and $y$ of $\mathbb{F}_q^N$, we denote their Schur product by
\[
x  \star  y \coloneqq (x_1 y_1, \ldots, x_N y_N).
\]
Given a linear code $C$, we can form the code $C  \star  C = \{\, c  \star  c' \in \mathbb{F}_q^N : c, c' \in C \,\}$, which we also denote by $C^{ \star 2}$. More generally, we define
\[
C^{ \star m} = \{\, c^1  \star  c^2  \star  \dots  \star  c^m : c^1, c^2, \dots, c^m \in C \,\}.
\]
Let $C \subseteq \mathbb{F}_q^{N}$ be a linear code, and let $\underline{u} = (u_1, \ldots, u_N) \in \mathbb{F}_q^N$ be a vector with only nonzero components, i.e., $u_1, \ldots, u_N \neq 0$. We define
\[
\underline{u}  \star  C \coloneqq \{\, (u_1 c_1, \ldots, u_N c_N) \in \mathbb{F}_q^{N} : (c_1, \ldots, c_N) \in C \,\}.
\]
It is clear that the codes $C$ and $\underline{u}  \star  C$ have the same length, dimension, and distance. Note, however, that the automorphism groups ${\rm Aut}(C)$ and ${\rm Aut}(\underline{u}  \star  C)$ are, in general, non-isomorphic.

\begin{definition}
Let $m \geq 1$ be an integer. We say that a linear code $C$ satisfies the \emph{$m$-multiplication property} if there exists a fixed nonzero vector $\underline{u}$ such that
\[
\underline{u}  \star  C^{ \star m} \subseteq C^\perp.
\]
\end{definition}

For any vector space $\mathbb{F}_q^N$, we call the vector $(1, 1, \ldots, 1)$, with all entries equal to one, the \emph{all-one vector}.

\begin{lemma}\label{lemma:multiplication-property}
Let $C$ be a code satisfying the $m$-multiplication property and containing the all-one vector. Then, $C$ satisfies the $\widetilde m$-multiplication property for every $1 \leq \widetilde m \leq m$.
\end{lemma}

\begin{proof}
By definition of the $m$-multiplication property, for any $m$ codewords $c^1, c^2, \dots, c^{m} \in C$ and any vector $c' \in C$, we have
\[
\langle \underline{u}  \star  c^1  \star  c^2  \star  \dots  \star  c^{m}, c' \rangle = 0.
\]
In particular, fixing the first $r$ codewords to be the all-one vector, we obtain
\[
\langle \underline{u}  \star  c^{r+1}  \star  \dots  \star  c^{m}, c' \rangle = 0.
\]
Thus, the code satisfies the $(m-r)$-multiplication property for every $r < m$.
\end{proof}

Notice moreover that $\underline{u}  \star  C \subseteq C^\perp$, and since the all-one vector belongs to $C$, it follows that $\underline{u} \in C^\perp$.

\subsection{Function fields and algebraic geometry codes}

In this section, we fix notation regarding algebraic function fields, as they play a key role in the construction of quantum codes supporting transversal non-Clifford gates. For supplemental definitions regarding the theory of function fields, we refer the reader to Stichtenoth~\cite{Stichtenothbook}.

\paragraph{Function fields.}
Let $\mathbb{F}_q(z)$ be the rational function field. Given an algebraic function field $F / \mathbb{F}_q(z)$, a \emph{place} $P$ of $F$ is the maximal ideal of some valuation ring of $F$. If $P$ is a place, we denote by $\mathcal{O}_P$ its valuation ring. Here, we follow the notation of Stichtenoth~\cite{Stichtenoth2006}; in other contexts, $P$ may be regarded as a point on a curve. The \emph{residue field} at $P$ is the quotient ring $\mathcal{O}_P / P$. The \emph{evaluation} of an element $f \in \mathcal{O}_P$ at $P$ is the image of $f$ under the canonical map $\mathcal{O}_P \to \mathcal{O}_P / P$. It is denoted by $f(P)$ and can be read as $f \bmod P$. The \emph{degree} of a place is the degree of the extension of its residue field. A place is called \emph{rational} if its residue field is isomorphic to $\mathbb{F}_q$, or equivalently, if its degree is equal to one.

Notably, the $q+1$ rational places of the rational function field $\mathbb{F}_q(z)$ correspond to the set $\mathbb{F}_q \cup \{\infty\}$. For $\alpha \in \mathbb{F}_q$, we denote by $(z = \alpha)$ the finite rational place corresponding to the maximal ideal generated by the function $(z - \alpha)$, and by $(z = \infty)$ the rational place corresponding to the point at infinity.

Given an algebraic extension $F / E$ such that $\mathbb{F}_q(z) \subseteq E \subseteq F$, and a place $Q$ of $F$, we say that $Q$ \emph{lies above} a place $P$ of $E$ if $P \subseteq Q$, and we denote this relation by $Q \mid P$. If $Q$ does not lie above $P$, we denote it as $Q \nmid P$. If $F / E$ is of finite degree, the number of places of $F$ lying above a given place of $E$ is at most the degree of the extension $[F : E]$. We say that a place $P$ of $E$ \emph{splits completely} in $F$ if the number of places $Q$ of $F$ such that $Q \mid P$ is exactly $[F : E]$.

\paragraph{Divisors and Riemann-Roch spaces.}
A \emph{divisor} $D$ in a function field $F / \mathbb{F}_q(z)$ is a formal sum of finitely many places with integer coefficients, namely $D = \sum n_P P$. A divisor $D$ is said to be \emph{positive} or \emph{effective}, written $D \geq 0$, if all of its coefficients are positive integers. Given an element $f \in F$, we write $(f)$ for its associated \emph{principal divisor}. Moreover, we denote the \emph{Riemann-Roch space} associated with a divisor $D$ as
\[
\mathcal{L}(D) = \{ f \in F : (f) + D \geq 0 \} \cup \{ 0 \}.
\]
Writing $D = D_{\text{pos}} - D_{\text{neg}}$ as the difference of two effective divisors, we see that $f \in \mathcal{L}(D)$ if and only if $f$ has poles of limited order at the places of $D_{\text{pos}}$ and zeros of sufficient order at the places of $D_{\text{neg}}$. If we have another divisor $D' \geq D$, then $\mathcal{L}(D')$ allows functions with fewer zeros and more poles, and thus $\mathcal{L}(D) \subseteq \mathcal{L}(D')$.

\paragraph{Towers of function fields.}
Let $\mathbb{F}_q$ be a field of characteristic $p$. A \emph{tower of function fields} over $\mathbb{F}_q$ is an infinite sequence $\mathcal{E} = (E_0, E_1, E_2, \ldots)$, where each $E_i$ is an algebraic function field of one variable over $\mathbb{F}_q$, satisfying the following property: for all $i \geq 0$, the extensions $E_{i+1} / E_i$ are separable of finite degree $[E_{i+1} : E_i] > 1$; $\mathbb{F}_q$ is the full constant field of each $E_i$; the genus $g(E_i)$ tends to infinity as $i \to \infty$.\par
Crucially for this work, a tower $\mathcal{E} = (E_0, E_1, E_2, \ldots)$ is called a \emph{Galois tower} if all extensions $E_i / E_0$ are Galois.

\paragraph{Automorphisms of function fields.} 
Given an algebraic extension of function fields $F/E$, an element $\sigma \in \text{Aut}(F /E)$ is a permutation of the elements of $F$ that fixes $E$. When the extension is of finite degree, we have that $|\text{Aut}(F/E)| \leq [F:E]$, with equality for extensions referred to as \textit{Galois extensions}.\par

The action of $\text{Aut}(F /E)$ extends to valuation rings and places. We denote $\sigma(P) \coloneqq \{\sigma(f) : f \in P\}$ and $\mathcal O_{\sigma(P)} \coloneqq \{\sigma(f) : f \in \mathcal O_P\}$. We indeed have that $\sigma(P)$ is the unique maximal ideal of the valuation ring $\mathcal O_{\sigma(P)}$, and moreover, $\sigma$ induces an isomorphism between the residue field of $P$ and that of $\sigma(P)$. Thus, $\text{Aut}(F /E)$ permutes the places of $F$. Importantly, if $P$ is a place of $F$ lying over a place $Q$ of $E$, then $\sigma(P)$ still lies over $Q$, since $\text{Aut}(F /E)$ fixes $E$.\par

Let $P$ be a place of degree one. Then, given an element $f \in \mathcal O_P$, we have the important relation
\begin{equation}\label{equation contravariance}
    f(\sigma(P)) = \sigma^{-1}(f)(P).
\end{equation}
In words, the evaluation of the function $f$ at the place $\sigma(P)$ is equal to the evaluation of the function $\sigma^{-1}(f)$ at $P$. Indeed, since $f(P) \in \mathbb F_q$, we have $f - f(P) \in P$, and thus $\sigma(f) - f(P) \in \sigma(P)$. We therefore obtain $\sigma(f)(\sigma(P)) = f(P)$. Substituting $f$ with $\sigma^{-1}(f)$ yields the relation.\par

The action of $\text{Aut}(F /E)$ extends naturally to an action on the divisor group, by setting, for $D = \sum n_P P$, $\sigma(D) = \sum n_P \sigma(P)$. A divisor is invariant under $\sigma$ if $\sigma(D) = D$.

\paragraph{Algebraic geometry codes.} 
Given a function field $F/\mathbb{F}_q$ and $N$ distinct rational places of $F$, written $P_1, \ldots, P_N$, we can form the divisor $R = P_1 + \ldots + P_N$. Let $H$ be a divisor with ${\rm supp}\, R \cap {\rm supp}\, H = \emptyset$. Moreover, let $\mathcal{L}(H) \subseteq F$ denote the Riemann–Roch space of the divisor $H$. We may consider the \textit{algebraic geometry (AG) code}
\[
C_{\mathcal{L}}(R, H) \coloneqq \{ (f(P_1), \ldots, f(P_N)) \in \mathbb{F}_q^{N} \mid f \in \mathcal{L}(H) \}.
\]
We next review an important property of AG codes that we will use throughout. Let $H$ and $H'$ be two divisors such that $H' \geq H$, i.e., all the coefficients of $H'$ are greater than those of $H$. Then, we have $\mathcal{L}(H) \subseteq \mathcal{L}(H')$ from our previous discussion on divisors and Riemann-Roch spaces. This means that, for AG codes, we have $C_\mathcal{L}(R, H) \subseteq C_\mathcal{L}(R, H')$.\par

What makes AG codes extremely useful is the existence of lower bounds on their parameters arising directly from the Riemann–Roch theorem. A classical result is the following.
\begin{theorem}\label{Theorem standard estimate AG codes}
    Let $g(F)$ denote the genus of the function field $F$, defined over $\mathbb F_q$, and let $R$ and $H$ be defined as above. Suppose that
    \[
        2g(F) - 2 < \deg H < N.
    \]
Then, the code $C_\mathcal{L}(R,H)$ is a $[N, K, D]$-code such that $K = \deg H + 1 - g(F)$ and $D \geq N - \deg(H)$. Moreover, the dual code of $C_\mathcal{L}(R,H)$ has distance $D^\perp \geq \deg(H) - 2g(F)$.
\end{theorem}
Of importance to us is the automorphism group of an AG code. Suppose that $\text{Aut}(F /E)$ fixes two given divisors $R$ and $H$. Then, the Riemann–Roch space $\mathcal{L}(H)$ is also invariant under the action of $\text{Aut}(F /E)$. This yields the following result.

\begin{proposition}\label{proposition AG code automorphism}
Let $F/\mathbb F_q(z)$ be a function field, and let $C_\mathcal{L}(R,H)$ be an AG code such that $\text{Aut}(F /E)$ fixes two divisors $R$ and $H$. Then $\text{Aut}(F /E)$ acts on $C_\mathcal{L}(R,H)$ by setting, for $f \in \mathcal{L}(H)$ and $\sigma \in \text{Aut}(F /E)$,
\[
\sigma\big((f(P_1), \ldots, f(P_N))\big) = \big(f(\sigma(P_1)), \ldots, f(\sigma(P_N))\big).
\]
Moreover, if $N > 2g(F) + 2$, this action is injective, and $\text{Aut}(F /E)$ can be seen as a subgroup of $\text{Aut}(C_\mathcal{L}(R,H))$.
\end{proposition}

When we later combine automorphisms of a code with their multiplicative properties, it will become crucial to notice that
\begin{equation}
    (f(\sigma(P_1)), \ldots, f(\sigma(P_N))) = (\sigma^{-1}(f)(P_1), \ldots, \sigma^{-1}(f)(P_N)).
\end{equation}
This follows directly from our previous discussion on the action of the automorphism group of a function field, and in particular from Equation~\ref{equation contravariance}.\par

\subsection{Asymptotically good family of transitive AG codes with multiplication property  }

In this work, we leverage the optimal Galois tower devised by Stichtenoth \cite{Stichtenoth2006}. We first review its main characteristics, which will later be used to define the AG codes of \cite{Stichtenoth2006} and to lower bound their parameters.

\begin{theorem}[{{\cite[Theorem 1.7.]{Stichtenoth2006}}}] \label{Theorem tower}
Let $\mathbb F_q$ be a field of characteristic $p$, such that $q$ is a square, i.e.\ $q = \ell^2$. Then, there exists an infinite tower $\mathcal E = (E_0 \subseteq E_1 \subseteq E_2 \subseteq \ldots)$ of function fields $E_i/\mathbb{F}_q$ satisfying the following properties:
\begin{enumerate}
\item $E_0 = \mathbb{F}_q (z)$ is the rational function field.
\item There exists an element $w \in E_1$ such that $w^{\ell -1} = z$. Hence, we have $E_0 \subseteq \mathbb{F}_q (w) \subseteq E_1$, and the extension $\mathbb{F}_q (w) / E_0$ is of degree $(\ell -1)$.
\item All extensions $E_i / E_0$ are Galois, and the degree of $E_i / E_0$ is $N_i \coloneqq [E_i : E_0] = (\ell -1)\ell^n p^{t(i)},$ where $t(i)$ is a non-negative integer.
\item The place $(z = 1)$ of $E_0$ splits completely in all extensions $E_i / E_0$. Namely, there are $[E_i : E_0]$ distinct places of $E_i$ above $(z = 1)$, and all of them are rational places of $E_i$.
\item The principal divisor of the function $w$ in $E_i$ is of the form $(w)^{E_i} = e^{(i)}_{0} A^{(i)} - e^{(i)}_{\infty} B^{(i)}$, where $A^{(i)} > 0$ and $B^{(i)} > 0$ are effective divisors of $E_i$. The integer $e^{(i)}_{0}$ is the ramification index of the place $(w = 0)$ in $E_i / \mathbb{F}_q (w)$ and has the form $e^{(i)}_{0} = \ell^{n-1} p^{r(i)}, \quad r(i) \geq 0$. The integer $e^{(i)}_{\infty}$ is the ramification index of the place $(w = \infty)$ in $E_i / \mathbb{F}_q (w)$ and has the form $e^{(i)}_{\infty} = \ell^n p^{s(i)},$ where $s(i) \geq 0.$
\item We have the following relation $e_0^{(i)}\deg A^{(i)} = e_\infty^{(i)}\deg B^{(i)} = [E_i:\mathbb{F}_q(w)]$.
\item The genus $g(E_i)$ satisfies
\[
g(E_i) = [E_i : \mathbb F_q (w)] + 1 - \big(\deg A^{(i)} + \deg B^{(i)}\big) \leq [E_i : \mathbb F_q (w)].
\]
\end{enumerate}
\end{theorem}
In his work, Stichtenoth \cite{Stichtenoth2006} proves that this tower of function fields is optimal, namely it attains the Drinfeld–Vladut bound $q^{1/2} - 1$.\par

\paragraph{Stichtenoth's AG code construction.}
We now focus on the family of AG codes constructed by Stichtenoth \cite{Stichtenoth2006}. These are associated with the function fields in the tower $\mathcal E = (E_0, E_1, E_2, \ldots)$ of Theorem~\ref{Theorem tower}. In the following definition, we use the notations of that theorem.

\begin{definition}\label{definition Stichtenoth codes}
Let $R^{(i)}$ be the divisor corresponding to the sum of all the places of $E_i$ lying over $(z = 1)$, namely $R^{(i)} = \sum_{P \mid (z = 1)} P$. Let $a_i\coloneqq  2( e_0^{(i)}-1)$ and $b_i\coloneqq(\ell-1) e_\infty^{(i)}-2$. For integers $a,b$ with $0 \leq a \leq a_i$ and $0 \leq b \leq b_i$, let $H^{(i)} = a A^{(i)} + b B^{(i)}$ be a divisor of $E_i$. The code $C^{(i)}_{a,b}$ of \cite{Stichtenoth2006} is defined by
\[
C^{(i)}_{a,b} \coloneqq C_{\mathcal{L}}(R^{(i)}, H^{(i)}).
\]
\end{definition}
To learn about the origin of the integers $a_i,b_i$, see the proof of Proposition 4.4 in \cite{Stichtenoth2006}.
For our purposes, there are three key characteristics of these codes that may all be fulfilled in certain parameter regimes: they may be transitive, asymptotically good, and satisfy the multiplication property.

\paragraph{Transitive group action.}

\begin{lemma}
For any $i$ and any integers $a,b$ with $0 \leq a \leq a_i$ and $0 \leq b \leq b_i$, the codes $C^{(i)}_{a,b}$ are transitive.
\end{lemma}

To establish this property of the codes $C^{(i)}_{a,b}$, notice that the places $P_1, \ldots, P_N$ in the support of the divisor $R^{(i)}$ are precisely the places of $E_i$ lying above $(z = 1)$, which splits completely. The Galois group ${\rm Gal}(E_i / E_0)$ therefore acts transitively on the support of $R^{(i)}$. It can also be shown that the divisors $A^{(i)}$ and $B^{(i)}$ are invariant under ${\rm Gal}(E_i / E_0)$, so that the Galois group leaves the divisor $H^{(i)}$ invariant as well.

\paragraph{Multiplication property.} 
A central aspect of the family of codes devised by Stichtenoth \cite{Stichtenoth2006} that we shall leverage for our quantum code construction is their multiplication property. 

\begin{proposition}\label{proposition dual AG code}
There exists a vector $\underline{u}$ with only non-zero entries such that
\[
(C^{(i)}_{a,b})^{\perp} = \underline{u}  \star  C^{(i)}_{a_i - a, b_i - b}.
\]
\end{proposition}
The nature of the vector $\underline{u}$ is not discussed here, as it involves the manipulation of differentials of function fields. Nevertheless, this result may be used to prove the following lemma.

\begin{lemma}\label{lemma multiplication property Stichtenoth}
Let $s\geq 2$ be an integer and let $m=s-1$. Suppose $a_i, b_i \geq s$, and fix the parameters $a = \left\lfloor \frac{a_i}{s} \right\rfloor$ and $b = \left\lfloor \frac{b_i}{s} \right\rfloor$. Then, for any $1 \leq \widetilde m \leq m$, the code $C^{(i)}_{a,b}$ satisfies
\[
\underline{u}  \star  \big(C^{(i)}_{a,b}\big)^{ \star \widetilde m} \subseteq (C^{(i)}_{a,b})^\perp.
\]
In words, the code $C^{(i)}_{a,b}$ has the $\widetilde m$-multiplication property for every $\widetilde m \leq m$.
\end{lemma}

We defer the proof of this lemma to Appendix~\ref{Appendix deferred proofs on code parameters}, for readability.  

\paragraph{Asymptotic performance.}
We now estimate the asymptotic performance of the family of codes $C^{(i)}_{a,b} \coloneqq C_{\mathcal{L}}(R^{(i)}, H^{(i)})$, $i \in \mathbb N$.

\begin{proposition}\label{proposition classical code parameters}
Let $\ell, s$ be positive integers such that $s\geq 2$ and $\ell \geq 2 s-1 $, and let $q = \ell^2$. Let the integers $a,b$ be such that $a = \left\lfloor \frac{a_i}{s} \right\rfloor$ and $b = \left\lfloor \frac{b_i}{s} \right\rfloor$. Then, the code $C^{(i)}_{a,b}$ is an $[N_i, K_i, D_i]_q$ code over $\mathbb F_q$ such that $2g(E_i) - 2 < \deg H^{(i)} < N_i$. Hence,
\begin{align*}
    K_i&\geq  \left( \frac{\ell+1-s}{s(\ell-1)}\right)N_i,\\
    D_i &\geq \left (1-\frac{3}{s(\ell-1)}\right)N_i,\\
   D_i^\perp  &\geq\left( \frac{\ell+1}{s(\ell-1)}\right)N_i.
\end{align*}
In particular, the sequence $(C^{(i)}_{a,b})_{i\geq 1}$ forms an asymptotically good family of codes.
\end{proposition}

For conciseness, the proof is deferred to Appendix~\ref{Appendix deferred proofs on code parameters}, where we show that the degree of the divisor $H^{(i)}$ is bounded as claimed. The results on the dimension and distance of the codes then follow directly from Theorem~\ref{Theorem standard estimate AG codes}.
For simplicity, we express the code parameters in terms of $s$, but notice that when $s=m+1$, the code satisfies the $m$-multiplication property.

\begin{remark}
Crucially, \cite{he2025good} fixed the parameters $a = \left\lfloor \frac{a_i}{4} \right\rfloor$ and $b = \left\lfloor \frac{b_i}{4} \right\rfloor$, and proved that the code family $(C^{(i)}_{a,b})_{i\in \mathbb N}$ is asymptotically good and each member satisfies the $ m$-multiplication property for any $ m \leq 3$. This was key to show that logical $\mathsf{CCZ}$ gates may be applied on every possible triple of logical qubits. Later, we take $s\geq 4$ and show that many of these gates can be performed in parallel.
\end{remark}

\section{Good quantum codes with addressable and parallelizable logical $\mathsf{C}^{m-1}Z$ Gates}\label{Good quantum codes with addressable and parallelizable}

\subsection{Good quantum code with multiplication property}

Given the code $C_{a,b}^{(i)}$ of dimension $K_i$, a generator matrix is a full-rank $K_i \times N_i$ matrix such that every row represents the valuation of a basis element (a function) in the Riemann-Roch space $\mathcal L(H^{(i)})$. 

Now, let us choose any function field $F$ in the tower $\mathcal E = E_0 \subseteq \mathbb{F}_q(w) \subseteq F_0 \subseteq E_1 \subseteq E_2 \subseteq \dots$, satisfying $F_0 \subseteq F \subseteq E_i$, where all extensions are Galois and the intermediate extension $F_0\coloneqq \mathbb F_q(x_0)$ is a degree $\ell$ Galois extension of $\mathbb F_q(w)$ defined by the polynomial $w=x_0^\ell +x_0$. We may choose an arbitrary rational place $\zeta$ of $F$ such that $\zeta \mid (z=1)$. Then, we partition the places of $E_i$ lying over $(z=1)$ into the sets $\{Q : Q \mid \zeta\}$ and $\{P : P \nmid \zeta\}$. Let us denote 
\begin{equation}
    k_i \coloneqq |\{Q : Q \mid \zeta\}| = [E_i : F].
\end{equation}
Constructing a quantum code from a classical code with the multiplication property requires expressing a generator matrix of the classical code in a specific form, as made explicit in the following proposition, which is an extension of the approach of \cite{he2025good}.

\begin{proposition}\label{proposition quantum code punctured matrix}
Let $\ell\geq 4$ be a prime power and let $s\geq 1$ be an integer satisfying $\ell \geq 2s-1$. Let $\zeta$ be a rational place of $F$ such that $\zeta \mid (z=1)$. Suppose $a_i, b_i \geq s$, and fix the parameters $a = \left\lfloor \frac{a_i}{s} \right\rfloor$ and $b = \left\lfloor \frac{b_i}{s} \right\rfloor$. Then, $D_i^\perp >k_i$ and $C_{a,b}^{(i)}$ admits a generator matrix of the form
\vspace{1.5em}
\begin{equation}
\tilde G =
\begin{pNiceMatrix}[columns-width=2.7em]
  I_{k_i} & G_1 \\
  0       & G_0 \\
  \CodeAfter
    \OverBrace[shorten,yshift=3pt]{1-1}{2-1}{Q \mid \zeta}
    \OverBrace[shorten,yshift=3pt]{1-2}{2-2}{P \nmid \zeta}
\end{pNiceMatrix}
\end{equation}
where each $Q \mid \zeta$ and each $P \nmid \zeta$ indexes a column of $\tilde G$, and the block $I_{k_i}$ is the $k_i \times k_i$ identity matrix.
\end{proposition}

For completeness, the proof of this proposition is deferred to Appendix~\ref{Appendix deferred proofs on quantum code parameters}, as it slightly generalizes the construction of \cite{he2025good}. Here, we allow more freedom in the choice of the place $\zeta$, which provides greater control over the dimension-distance tradeoff of the resulting quantum code.

From this decomposition, we see that any of the first $k_i$ rows of $\tilde G$ corresponds to a function $\tilde g_Q$, where $Q \mid \zeta$. Moreover, evaluating this function on any place $Q' \mid \zeta$ yields
\begin{equation}\label{equation delta function g}
    \tilde g_Q(Q') = \delta_{Q,Q'}.
\end{equation}
To construct the quantum code, we use the framework developed by Krishna and Tillich \cite{Krishna2019} and puncture the matrix $\tilde G$ on the first $k_i$ coordinates. This produces the matrix
\begin{equation}
G = \begin{pmatrix} G_1 \\ G_0 \end{pmatrix}.
\end{equation}
We can now define two codes: $C_0 = \ker G_0$, and $C_1 = \text{rowspan}(G)$. We observe that $C_0^\perp = \text{rowspan}(G_0) \subseteq C_1$, and it is thus possible to define the quantum code
\begin{equation}\label{equation quantum code def}
    \mathcal Q^{(i)}_{\zeta} = \text{CSS}(C_0, C_1),
\end{equation}
where $C_0$ corresponds to the logical $Z$ operators and $C_1$ to the logical $X$ operators. The physical qubits are in one-to-one correspondence with the set $\{P : P \nmid \zeta\}$, which now serves as an index set for the physical qubits.\par

Before discussing the logical operators, let us first observe that $\text{rowspan}(G_0) \cap \text{rowspan}(G_1) = 0$ and that $G_1$ is full rank. This statement is restated in Lemma \ref{lemma generator matrix independent blocks}, in Appendix ~\ref{Appendix deferred proofs on quantum code parameters}, along with a proof, for completeness. A consequence is that the logical qubits are in one-to-one correspondence with the set $\{Q : Q \mid \zeta\}$. Indeed, the non-trivial logical $X$-operators belong to $C_1 \setminus C_0^\perp$, and a basis for these operators is given by the rows of $G_1$. Since the first $k_i$ rows of $\tilde G$ are the evaluations of the functions $\tilde g_Q$, the corresponding rows of $G_1$ can be written as vectors $g_Q$, obtained by puncturing the first $k_i$ coordinates of $\tilde g_Q$. Hence, there is a one-to-one correspondence between the basis of logical operators (rows of $G_1$) and the functions $\tilde g_Q$ with $Q \mid \zeta$.

The number of logical qubits is thus
\begin{equation}\label{equation dimension quantum}
k_i = |\{Q : Q \mid \zeta\}| = [E_i : F],
\end{equation}
and the number of physical qubits is
\begin{equation}
n_i = N_i - k_i.
\end{equation}

The next proposition introduces the standard logical orthonormal basis from CSS code theory.

\begin{proposition}\label{proposition standard basis CSS}
Let the set $(g_Q)_{Q \mid \zeta}$ correspond to the rows of $G_1$, and let $x = (x_Q)_{Q \mid \zeta} \in \mathbb F_q^{k_i}$. An orthonormal basis for $\mathcal Q^{(i)}_{\zeta}$ is given by
\[
\left\{ \overline{\ket{x}} \coloneqq \sum_{g \in \text{rowspan}(G_0)} \ket{ \sum_{Q \mid \zeta} x_Q g_Q + g} : x \in \mathbb F_q^{k_i} \right\},
\]
up to state normalization.
\end{proposition}

It remains to study the parameters of the family of quantum codes obtained from this procedure. In \cite{he2025good}, the author optimized the rate of the family of quantum codes by fixing $\zeta$ to be a place of $F_0$ in the tower
\[
E_0 \subseteq \mathbb{F}_q(w) \subseteq F_0 \subseteq E_1,
\]
and by partitioning the places of $E_i$ between those lying over an arbitrary fixed place $\zeta$ of $F_0$ and those lying over a different place of $F_0$. More precisely, let $q = \ell^2$ be a fixed power of two such that $\ell \geq 8$, and let $a = \left\lfloor \frac{a_i}{4} \right\rfloor$, and $b = \left\lfloor \frac{b_i}{4} \right\rfloor$. Then, the code $C_{a,b}^{(i)}$ is asymptotically good and satisfies the $3$-multiplication property. Let $\zeta$ be a place of $F_0$ such that $\zeta \mid (z=1)$. Moreover, let $\mathcal Q^{(i)}_\zeta$ be defined as above, by partitioning the places of $E_i$ into the set 
\[
\{Q : Q \mid \zeta\} \quad (\text{logical qubits}) \quad \text{and} \quad \{P : P \nmid \zeta\} \quad (\text{physical qubits}).
\] The author of \cite{he2025good} proved that the family of codes $\mathcal Q^{(i)}_\zeta$ is asymptotically good.\par

Here, we take a slightly more general path, as we aim to prove that an asymptotically good family of quantum codes can be obtained from classical codes with the $m$-multiplication property for $m \geq 3$. We may use Lemma~\ref{lemma generator matrix independent blocks} to derive the following behavior of the parameters of the family of quantum codes.

\begin{theorem}\label{theorem quantum code parameters}
Let $\ell\geq 6$ be a prime power and let $s\geq 3$ be an integer such that $\ell \geq 2s$. Let $\zeta$ be a place of $F$, such that $F \in \mathcal E$ and $F \subset E_i$, with $\zeta \mid (z=1)$. Then, the code $\mathcal Q^{(i)}_{\zeta}$ defined in Equation \eqref{equation quantum code def} has dimension 
\[
k_i = [E_i : F]
\]
and distance 
\[
d_i \geq \min(D_i, D_i^\perp) - k_i.
\] 
In particular, the family of codes $(\mathcal Q^{(i)}_{\zeta})_{i\in \mathbb N}$ is asymptotically good. 
\end{theorem}

For conciseness, the proof is deferred to Appendix~\ref{Appendix deferred proofs on quantum code parameters}. Notice that here, we use the stricter bound $\ell \geq 2s$, compared to that of Proposition \ref{proposition quantum code punctured matrix} which has $\ell \geq 2s-1$.

\subsection{Parallelizable, addressable logical $\mathsf{C}^{m-1}\mathsf{Z}$ gates}

In this section, we prove that the quantum code $\mathcal Q_\zeta^{(i)}$ of Equation~\ref{equation quantum code def} supports transversal logical multi-control-$Z$ gates that are both addressable and parallelizable, to a certain extent. Let us emphasize that in~\cite{he2025good}, it was proven that for any $\gamma\in \mathbb F_q$, and any triple of logical qudits $Q_0,Q_1,Q_2$ belonging to three different code blocks of $\mathcal Q^{\otimes 3}$, the logical gate $\overline{\mathsf{CCZ}^\gamma[Q_0,Q_1,Q_2]}$ can be implemented by a circuit of physical $\mathsf{CCZ}$ gates of depth~1. 

Here, we extend this result by analyzing the parallelizability of generalized multi-control-$Z$ gates. To this end, let us recall that $\zeta$ is a place of a chosen function field $F\in \mathcal E$ where $F\subset E_i$ and $E_i/F$ is a Galois extension. A \textit{physical} $\mathsf{C}^{m-1}\mathsf{Z}$ gate on qudits corresponding to $P_0,P_1,\dots,P_{m-1}\nmid \zeta$ respectively in code blocks $0,1, 2,\dots,m-1$, is denoted $\mathsf{C}^{m-1}\mathsf{Z}^\gamma[P_0,P_1,\dots,P_{m-1}]$. In addition, a \textit{logical} $\mathsf{C}^{m-1}\mathsf{Z}$ gate on logical qudits $Q_0,Q_1,\dots,Q_{m-1}\mid \zeta$ belonging respectively to code blocks $0,1, 2,\dots,m-1$, is denoted by $\overline{\mathsf{C}^{m-1}\mathsf{Z}^\gamma[Q_0, Q_1,\dots, Q_{m-1}]}$.\par

Recall that for the AG code $C_{a,b}^{(i)}$, the first $k_i$ rows of a generator matrix $\tilde G$, as given in Proposition \ref{proposition quantum code punctured matrix}, are written $\tilde g_Q$ where $Q \mid \zeta$. We now introduce the following notion of \textit{modulation function}, which is central to the extension of the results of~\cite{he2025good}. 

\begin{definition}
Given any set of logical qudits corresponding to a subset $S\subseteq \{Q:Q\mid\zeta\}$, and a vector $\gamma=(\gamma_Q)_{Q\mid \zeta}\in \mathbb F_q^{k_i}$, we define the associated \emph{modulation function} $M$ as
\[
M \coloneqq \sum_{Q\in S}\gamma_Q  (\underline{u}_{Q})^{-1}  \tilde g_Q,
\]
where $\underline{u}=(\underline{u}_P)_{P\in \text{supp}(R^{(i)})}$ is the vector of Proposition \ref{proposition dual AG code}.
\end{definition}

Note that $M$ belongs to the Riemann-Roch space $\mathcal L(H^{(i)})$, since it is defined as a linear combination of functions in $\mathcal L(H^{(i)})$. Using this modulation function, we can define logical multi-control-$Z$ gates that can address specific logical qudits in parallel.

\begin{theorem} \label{theorem: main circuit}
Suppose $a_i, b_i \geq m+1$, and fix the parameters $a = \left\lfloor \frac{a_i}{m+1} \right\rfloor$ and $b = \left\lfloor \frac{b_i}{m+1} \right\rfloor$. Let $\mathcal{Q}^{(i)}_\zeta$ be the quantum code constructed from a generator matrix of $C_{a,b}^{(i)}$, as defined in Equation \eqref{equation quantum code def}. Let $S$ be a set of places $Q$ such that $Q|\zeta$. Then, for any logical state $\ket{\psi} \in (\mathcal{Q}^{(i)}_\zeta)^{\otimes m}$, given any modulation function
\[
M = \sum_{Q\in S}\gamma_Q  (\underline{u}_{Q})^{-1}  \tilde g_Q,
\]
and any $\sigma_1,\sigma_2,\dots,\sigma_{m-1}\in \text{Gal}(E_i/F)$, it holds that
\begin{align}\label{Equation physical circuit}
\prod_{Q\in S}\overline{\mathsf{C}^{m-1}\mathsf{Z}^{\gamma_Q}[Q,\sigma_1(Q),\sigma_2(Q),\dots,\sigma_{m-1}(Q)]} \ket{\psi}
=  
\prod_{P\nmid \zeta}\mathsf{C}^{m-1}\mathsf{Z}^{-\underline{u}_{P} M(P)}[P,\sigma_1(P),\sigma_2(P),\dots,\sigma_{m-1}(P)] \ket{\psi}.
\end{align}
\end{theorem}
For clarity, the proof of this theorem is deferred to Appendix~\ref{Appendix proof of main theorem}. It relies on the following direct corollary of Lemma~\ref{lemma multiplication property Stichtenoth}.

\begin{corollary}\label{Corollary multiplication property}
Fixing the parameters $a = \left\lfloor \frac{a_i}{m+1} \right\rfloor$ and $b = \left\lfloor \frac{b_i}{m+1} \right\rfloor$, the code $C_{a,b}^{(i)}$ has the $m$-multiplication property with the vector $\underline{u}$ of Proposition~\ref{proposition dual AG code}. That is, for any $m+1$ functions $\tilde f^0,\dots,\tilde f^{m}\in \mathcal L(H^{(i)})$, we have
\begin{equation}
\sum_{P\nmid \zeta}\underline{u}_P \tilde f^0(P)\tilde f^1(P)\dots \tilde f^{m}(P)
= - \sum_{Q|\zeta}\underline{u}_{Q}\tilde f^0(Q)\tilde f^1(Q)\dots \tilde f^{m}(Q).
\end{equation}
\end{corollary}
Using that the code $C_{a,b}^{(i)}$ also has the $\widetilde m$-multiplication for $1<\widetilde m\leq m$, the result of Theorem \ref{theorem: main circuit} can be extended straightforwardly to any product of $(\widetilde m-1)$-control-$Z$ gates.

\subsection{Circuit advantage}

In this section, we estimate the advantage of using a circuit of logical multi-control-$Z$ gates for which subsets of logical qudits can be addressed in parallel, using the quantum code of Equation \eqref{equation quantum code def} and the multi-control-$Z$ gates of Theorem~\ref{theorem: main circuit}.

Given a circuit $\mathcal C$ of multi-control-$Z$ gates of depth $L$ acting on $k$ ideal qudits, and a fault-tolerant circuit $\bar{\mathcal C}$ of depth $\bar L$ simulating $\mathcal C$ on the logical space of a quantum code of dimension $k$, we recall that the \emph{depth overhead} of $\bar{\mathcal C}$ is defined as the ratio $\bar L / L$. We now prove the following result.

\begin{theorem}\label{Theorem depth scaling}
Let $a_i,b_i,m\geq 1$ be integers such that $a_i, b_i \geq m+1$, and let $a = \left\lfloor \frac{a_i}{m+1} \right\rfloor$ and $b = \left\lfloor \frac{b_i}{m+1} \right\rfloor$. Let $\mathcal{Q}^{(i)}_\zeta$ be the quantum code constructed by puncturing a generator matrix of $C_{a,b}^{(i)}$, as defined in Equation \eqref{equation quantum code def}. Then, the depth of any logical circuit $\bar{\mathcal C}$ of $\overline{\mathsf{C}^{m-1}\mathsf{Z}}$ gates on the quantum code $(\mathcal Q^{(i)}_{\zeta})^{\otimes m}$ is upper bounded by $k_i^{m-1}$.
\end{theorem}

This should be compared with the all-to-all inter-block $\mathsf{C}^{m-1}\mathsf{Z}$ circuit on $m$ blocks of $k_i$ physical qudits, which has a minimal depth of exactly $k_i^{m-1}$, leading to a constant depth overhead for the all-to-all logical interblock $\overline{\mathsf{C}^{m-1}\mathsf{Z}}$ circuit.\par
Fix $s\coloneqq m+1\geq3 $ and let $\ell \geq 6$ be a prime power such that $\ell\geq 2s$. Combining Theorem \ref{theorem quantum code parameters}, with $a = \left\lfloor \frac{a_i}{s} \right\rfloor$ and $b = \left\lfloor \frac{b_i}{s} \right\rfloor$, and Theorem \ref{Theorem depth scaling}, we directly obtain a proof of Theorem \ref{Theorem, main} for circuits of $\overline{\mathsf{C}^{m-1}\mathsf{Z}}$ gates. Given any $1< \widetilde m\leq m$, it is immediate to see that the result also holds for circuits of $(\widetilde m-1)$-control-$Z$ gates on sub-codes of the form $(\mathcal Q^{(i)}_{\zeta})^{\otimes \widetilde m}$, with the depth upper bound $O(k_i^{\widetilde m-1})$. These depths can also be lowered by optimizing specific multi-control-$Z$ circuits.\par

Throughout, we assume that the conditions of Theorem \ref{Theorem depth scaling} are satisfied. We now make the following observations, clarifying the advantages of applying the circuit of Theorem~\ref{theorem: main circuit} to $m$ blocks of the quantum code $\mathcal Q^{(i)}_{\zeta}$.

\begin{lemma}
The circuit of Theorem~\ref{theorem: main circuit} is of depth one and can target any $m$-tuple of logical qudits, each belonging to a distinct code block of $(\mathcal Q^{(i)}_{\zeta})^{\otimes m}$.
\end{lemma}

\begin{proof}
For the first part of the lemma, given $P, P' \nmid \zeta$ and $\sigma \in \text{Gal}(E_i/F)$, if $\sigma(P) = \sigma(P')$, then $P = P'$.

For the second part, given any pair of places $Q,Q'|\zeta$, there exists an element $\sigma \in \text{Gal}(E_i/F)$ such that $Q' = \sigma(Q)$. This follows from the transitivity of $\text{Gal}(E_i/F)$ on the places above any given rational place of $F$ lying above the place $(z=1)$ in $\mathbb F_q(z)$. Accordingly, if we let $Q_0,Q_1,\dots,Q_{m-1}$ be any tuple of places and $\sigma_1,\dots,\sigma_{m-1}$ the corresponding elements such that $\sigma_i(Q_0) = Q_i$, it follows that
\[
\overline{\mathsf{C}^{m-1}\mathsf{Z}^{\gamma_Q}[Q_0,Q_1,\dots,Q_{m-1}]} 
= 
\overline{\mathsf{C}^{m-1}\mathsf{Z}^{\gamma_Q}[Q,\sigma_1(Q),\dots,\sigma_{m-1}(Q)]}.
\]
\end{proof}
Let $\Gamma$ denote the Cartesian product of $(m-1)$ copies of the group $\text{Gal}(E_i/F)$:
\[
\Gamma \coloneqq \text{Gal}(E_i/F) \times \text{Gal}(E_i/F) \times \dots \times \text{Gal}(E_i/F),
\]
and let an element be written in bold as $\boldsymbol{\sigma} = (\boldsymbol{\sigma}_i)_{i\in [m-1]} \in \Gamma$. For each $\boldsymbol{\sigma}\in \Gamma$, we can associate a set of logical qudits $S^{\boldsymbol{\sigma}}$ (a set of places, each referring to a logical qudit), a vector $\gamma^{\boldsymbol{\sigma}}$, and a modulation function 
    \[
    M^{\boldsymbol{\sigma}} \coloneqq \sum_{Q\in S^{\boldsymbol{\sigma}}} \gamma^{\boldsymbol{\sigma}}_Q (\underline{u}_Q)^{-1} \tilde g_Q.
    \]
We then obtain the following result.

\begin{proposition}\label{proposition logical circuit}
For any logical circuit $\bar{\mathcal C}$ of inter-block multi-control-$Z$ gates, there exists a subset $\overline{\Gamma}\subseteq \Gamma$, sets $S^{\boldsymbol{\sigma}}$, and vectors $\gamma^{\boldsymbol{\sigma}}$ associated to each $\boldsymbol{\sigma}\in\overline{\Gamma}$, such that
\[
\bar{\mathcal C} 
= 
\prod_{\boldsymbol{\sigma}\in\overline{\Gamma}} 
\prod_{Q\in S^{\boldsymbol{\sigma}}}
\overline{\mathsf{C}^{m-1}\mathsf{Z}^{\gamma^{\boldsymbol{\sigma}}_Q}[Q,\boldsymbol{\sigma}_1(Q),\boldsymbol{\sigma}_2(Q),\dots,\boldsymbol{\sigma}_{m-1}(Q)]},
\]
and $\bar{\mathcal C}$ can be implemented by a physical circuit of depth upper bounded by $|\overline{\Gamma}|$.
\end{proposition}

\begin{proof}
The logical circuit $\bar{\mathcal C}$ can be implemented by successively applying the physical circuit
\[
\prod_{P\nmid \zeta} 
\mathsf{C}^{m-1}\mathsf{Z}^{-\underline{u}_{P} M^{\boldsymbol{\sigma}}(P)}
[P,\boldsymbol{\sigma}_1(P),\boldsymbol{\sigma}_2(P),\dots,\boldsymbol{\sigma}_{m-1}(P)]
\]
of Theorem~\ref{theorem: main circuit}, which is of depth one. It follows that $\bar{\mathcal C}$ can be implemented by a circuit of depth $|\overline{\Gamma}|$.
\end{proof}

We are now ready for the proof of Theorem~\ref{Theorem depth scaling}.

\begin{proof}[Proof of Theorem~\ref{Theorem depth scaling}]
The depth of any logical circuit $\bar{\mathcal C}$ of inter-block multi-control-$Z$ gates on the quantum code $\mathcal Q^{(i)}_{\zeta}$ is upper bounded by $|\overline{\Gamma}|$, which is itself equal to 
\[
|\text{Gal}(E_i/F)|^{m-1} = k_i^{m-1}.
\]
\end{proof}

The bound of Theorem \ref{Theorem depth scaling} can be lowered with prior knowledge of the specific circuit ${\mathcal C}$ under consideration.\par

A dense multi-control-$Z$ circuit can be obtained from the all-to-all circuit, which has depth $k^{m-1}$, by removing a subset of its gates. However, as long as the average qubit connectivity (the number of $\mathsf{C}^{m-1}\mathsf{Z}$ gates per qubit) remains high, e.g. scaling linearly with $ k$, rescheduling the circuit cannot substantially reduce the depth. Therefore, in the regime of dense circuits, our depth bound $\Theta ( k^{m-1})$ is optimal.\par

In contrast, a sparse physical $\mathsf{C}^{m-1}\mathsf{Z}$ circuit $\mathcal C$, i.e one with low qubit connectivity, may be rescheduled to have a depth significantly lower than $k^{m-1}$. In that case, we may wonder whether a similar depth reduction is also possible for its logical implementation $\overline{\mathcal C}$. It turns out that the depth of $\overline{\mathcal C}$ can also be optimized by exploiting the code structure. For instance, given a circuit $\mathcal C$ of depth $L$ containing $k$ different $\mathsf{C}^{m-1}\mathsf{Z}$ gates, encoding induces a fault-tolerant circuit of the form $\bar{\mathcal C}=\prod_{\boldsymbol{\sigma}\in\overline{\Gamma}} \prod_{Q\in S^{\boldsymbol{\sigma}}} \overline{\mathsf{C}^{m-1}\mathsf{Z}^{\gamma^{\boldsymbol{\sigma}}_Q} [Q,\boldsymbol{\sigma}_1(Q), \boldsymbol{\sigma}_2(Q),\dots,\boldsymbol{\sigma}_{m-1}(Q)]}$, as stated in Proposition \ref{proposition logical circuit}. To minimize the depth of the circuit, one seeks to reduce the number of terms of the outer product by adapting the encoding to the specific connectivity pattern of $\mathcal C$.\par

For example, in the worst case there are $\Omega(k)$ different elements $\boldsymbol{\sigma}\in\overline{\Gamma}$ contributing to the outer products, yielding a logical circuit of depth $\Omega (k)$. However, for a large class of circuits with structured connectivity, it may be possible to adapt the encoding so that the number of elements  $\boldsymbol{\sigma}\in\overline{\Gamma}$ in the product becomes constant. Moreover, when a large circuit is composed of several routines of the form $\mathcal C$, each involving  a different combination of logical qubits, we may interleave successive routines with suitable $\mathsf{SWAP}$ operations to permute the logical qubit labels. This allows each routine to act on a locally favorable ordering, thereby reducing its depth overhead to a constant.

\newpage
\appendix

\section{Appendix}

\subsection{Deferred proofs on AG code properties}\label{Appendix deferred proofs on code parameters}

\begin{proof}[Proof of Lemma \ref{lemma multiplication property Stichtenoth}]
From our review of AG codes, we infer the following. Suppose we have natural numbers $a' \geq a$ and $b' \geq b$. Let $H'^{(i)} = a' A^{(i)} + b' B^{(i)}$. Then $H'^{(i)} \geq H^{(i)}$, and consequently
\[
C^{(i)}_{a,b} \subseteq C^{(i)}_{a',b'}.
\]
Moreover, given any two functions $f,g \in \mathcal{L}(H^{(i)})$, their product $fg$ belongs to $\mathcal{L}(2H^{(i)})$. It follows that
\[
(C^{(i)}_{a,b})^{ \star m} \subseteq C^{(i)}_{ma,mb} \subseteq C^{(i)}_{a_i - a,b_i - b},
\]
where the last inclusion follows from our choice of $a$ and $b$. Lastly, since $\underline{u}$ has only non-zero coordinates, we have
\[
\underline{u}  \star  (C^{(i)}_{a,b})^{ \star m} \subseteq \underline{u}  \star  C^{(i)}_{a_i - a,b_i - b} = (C^{(i)}_{a,b})^\perp,
\]
namely, the code has the $m$-multiplication property.

Furthermore, let $f = 1$ be the constant function. Since the divisor $H^{(i)}$ is positive, we have $(f) + H^{(i)} = H^{(i)} \geq 0$, and hence $f \in \mathcal{L}(H^{(i)})$. In other words, the all-one vector belongs to $C^{(i)}_{a,b}$. By Lemma~\ref{lemma:multiplication-property}, we conclude that $C^{(i)}_{a,b}$ has the $\widetilde m$-multiplication property for every $\widetilde m \leq m$.
\end{proof}

\begin{proof}[Proof of Proposition \ref{proposition classical code parameters}]
Let us first prove the degree bound $2g(E_i)-2< \deg H^{(i)}<N_i$. Indeed, for the upper bound, we have
\begin{align*}
    a\deg A^{(i)} + b\deg B^{(i)} &\leq \frac{a_i}{s}\deg A^{(i)} + \frac{b_i}{s}\deg B^{(i)}\\
    &\leq \frac{2(e_0^{(i)}-1)}{s}\deg A^{(i)} + \frac{(\ell-1)e_\infty^{(i)}-2}{s}\deg B^{(i)}\\
    &< \frac{2e_0^{(i)}}{s}\deg A^{(i)} + \frac{(\ell-1)e_\infty^{(i)}}{s}\deg B^{(i)}\\
    &\leq \frac{\ell +1}{s}[E_i:\mathbb F_q(w)] = \frac{(\ell+1)N_i}{s(\ell-1)} \\
    &\leq N_i,
\end{align*}
where in the fourth line we have used item \textit{6.} of Theorem \ref{Theorem tower}, namely that $e_0^{(i)}\deg A^{(i)} = e_\infty^{(i)}\deg B^{(i)} = [E_i:\mathbb{F}_q(w)]$. Moreover, since $\mathbb{F}_q(w)/E_0$ is an extension of degree $\ell-1$, we obtain $[E_i:\mathbb{F}_q(w)] = \frac{N_i}{\ell-1}$. In the fifth line, we have used that $s\geq 2$ and $\ell\geq 4$.\par

We now prove the lower inequality $2g(E_i)-2< \deg H^{(i)}$. First, note that
$$\deg H^{(i)} - 2g(E_i) + 2 = (a+2)\deg A^{(i)} + (b+2)\deg B^{(i)} - 2[E_i:\mathbb{F}_q(w)] ,$$
and therefore,
\begin{align*}
    \deg H^{(i)} - 2g(E_i) + 2 &> \left(\frac{a_i}{s}+1\right)\deg A^{(i)} + \left(\frac{b_i}{s}+1\right)\deg B^{(i)} - 2[E_i:\mathbb{F}_q(w)] \\
    &= \left(\frac{2(e_0^{(i)}-1)}{s}+1\right)\deg A^{(i)} + \left(\frac{(\ell-1)e_\infty^{(i)}-2}{s}+1\right)\deg B^{(i)} - 2[E_i:\mathbb{F}_q(w)] \\
    &= \left(\frac{2}{s}+\frac{\ell-1}{s}-2\right)[E_i:\mathbb F_q(w)] + \left(1-\frac{2}{s}\right)(\deg A^{(i)}+\deg B^{(i)})\\
    &> \left(\frac{\ell+1}{s}-2\right)[E_i:\mathbb F_q(w)],
\end{align*}
where in the third line we have again used that $e_0^{(i)}\deg A^{(i)}=e_\infty^{(i)}\deg B^{(i)}=[E_i:\mathbb F_q(w)]$, and in the fourth line that $\deg A^{(i)}, \deg B^{(i)}>0$. Thus, if $\ell\geq 2s-1$, it follows that $\deg H^{(i)} - 2g(E_i) + 2>0$.\par

We conclude on the parameters of the code using the standard estimates from Theorem \ref{Theorem standard estimate AG codes},
\begin{align*}
K_i &= \deg H^{(i)} + 1 - g(E_i), \\
D_i &\geq N_i - \deg H^{(i)},\\
D_i^\perp &\geq \deg H^{(i)} - 2g(E_i).
\end{align*}
Explicitly, this yields
\begin{align*}
    K_i &\geq \left(\frac{\ell+1}{s}-1\right)[E_i:\mathbb F_q(w)]\\
    &= \left(\frac{\ell+1-s}{s(\ell-1)}\right)N_i,
\end{align*}
and
\begin{align}
    D_i &\geq \left(1 - \frac{3}{s(\ell-1)}\right)N_i,
\end{align}
where we have used that $[E_i:\mathbb F_q(w)] = [E_i:E_0]/(\ell-1)$. Finally, for later use, it will be important to lower bound the distance of the dual code, which is obtained as follows:
\begin{align*}
    D_i^\perp &\geq \frac{\ell+1}{s}[E_i:\mathbb F_q(w)] + \left(1 - \frac{2}{s}\right)(\deg A^{(i)}+\deg B^{(i)})\\
    &> \frac{\ell+1}{s}[E_i:\mathbb F_q(w)]\\
    &= \left(\frac{\ell+1}{s(\ell-1)}\right)N_i.
\end{align*}
\end{proof}

\subsection{Deferred proof on quantum code parameters}\label{Appendix deferred proofs on quantum code parameters}

\begin{proof}[Proof of Proposition \ref{proposition quantum code punctured matrix}]
In this parameter regime, we have $D_i^\perp\geq \deg H^{(i)}-2g(E_i)$. Therefore, using the same arguments as in the proof of Proposition \ref{proposition classical code parameters}, we have that
\begin{align*}
    D_i^\perp > \frac{\ell+1}{s(\ell-1)}[E_i:E_0].
\end{align*}
Now, we have different choices for the location of the place $\zeta$: either it belongs to $F=F_0$, or it belongs to $F=E_v$, where $v<i$. In the former case, we have 
\begin{align*}
  \frac{D_i^\perp}{k_i} &> \frac{\ell+1}{s(\ell-1)}[F_0:E_0]\\
    &= \frac{\ell(\ell+1)}{s},
\end{align*}
and, using that $2s-1\leq \ell$, the right-hand side is greater than one. 

Now, if $\zeta$ is a place of $E_v$ for $1\leq v<i$, we always have $k_i< [E_i:F_0]$, so that $D_i^\perp/k_i$ satisfies the same lower bound. For completeness, we would have  $D_i^\perp/k_i>1$ whenever $s\leq  \ell^v(\ell+1)p^{t(v)}$, which is satisfied for our choice of $\ell$ and $s$. 

Consequently, given any generator matrix of the code, any $k_i$ of its columns are linearly independent. This ensures that the generator matrix we start with can be put in the form given in the proposition, where we have the freedom to choose the identity block to be located at the columns of the set of places lying over the place $\zeta$. 
\end{proof}

\begin{lemma}\label{lemma generator matrix independent blocks}
We have $\text{rowspan}(G_0) \cap \text{rowspan}(G_1) = 0$ and $G_1$ is full rank.
\end{lemma}

\begin{proof}[Proof of Lemma \ref{lemma generator matrix independent blocks}]
    This follows directly from the containment $\underline{u} \star C_{a,b}^{(i)}\subseteq C_{a,b}^{(i)}{^\perp}$. Given $g,g'\in \text{rows}(G)$, the inner product $\langle g;\underline{u} \star g'\rangle$ is non-zero if $g=g'\in  \text{rows}(G_1) $, and is equal to zero otherwise. 

Now, suppose that there exists a relation $\sum_{\text{rows}(G)} x_g g=0$. For any $g'\in \text{rows}(G_1)$, we have
    \[
        \langle \sum x_g g ; \underline{u} \star g' \rangle = \sum x_g \langle g ; \underline{u} \star g' \rangle = x_{g'}\langle g';\underline{u} \star g'\rangle = 0,
    \]
so that every coefficient $x_{g'}$, for $g'\in \text{rows}(G_1)$, vanishes. Other coefficients also vanish, since $G_0$ is full rank.
\end{proof}

\begin{proof}[Proof of Theorem \ref{theorem quantum code parameters}]
The dimension of the quantum code is given by the dimension of $C_1/C_0^\perp$, and because of Lemma \ref{lemma generator matrix independent blocks}, this is equal to the number of rows of $G_1$, which is $k_i=[E_i:F]$, as given in Equation \eqref{equation dimension quantum}. There are two cases to compute the rate of the code. Either $F=F_0$ or $F=E_v$ for some $1\leq v<i$. In the former case, we have a constant rate
\[
\frac{k_i}{n_i}\geq\frac{k_i}{N_i}=\frac{1}{\ell(\ell-1)}.
\]
In the latter case, we obtain similarly 
\[
\frac{k_i}{n_i}\geq\frac{1}{(\ell-1)\ell^vp^{t(v)}},
\]
which is a constant since $v$ is fixed for the whole family.\par

Regarding the distance, we have $d_i\geq \min(d(C_0),d(C_1))$. Puncturing one position of the generator matrix of any code has the effect of decreasing the distance by at most one or leaving it unchanged. It follows that 
\[
d(C_1)\geq D_i-k_i,
\]
where $D_i$ is the distance of $C_{a,b}^{(i)}$.\par

To determine the distance of the code $C_0$, notice first that $G$ is a generator matrix for the AG code $C_\mathcal{L}(R'^{(i)},H^{(i)})$, where $R'^{(i)}= \sum_{P\nmid \zeta} P$, namely we consider the evaluations at the places that do not lie over $\zeta$. Moreover, the block $G_0$ describes evaluation of functions of the Riemann-Roch space $\mathcal{L}(H^{(i)})$ with the additional constraint that they vanish at every place $Q$ lying above $\zeta$. We can deduce that 
\[
C_0^\perp = C_\mathcal{L}(R'^{(i)},H'^{(i)}),
\]
where we define the divisor $H'^{(i)} \coloneqq H^{(i)} - \sum_{Q|\zeta} Q$. 

To compute a lower bound on the distance of $C_0$ and its dual, we must check that the conditions of Theorem \ref{Theorem standard estimate AG codes} are fulfilled. First, recalling that $\text{deg }H^{(i)}<N_i$, we have 
\[
\text{deg }H'^{(i)}<N_i-k_i,
\]
where the r.h.s is the number of places in the support of $R'^{(i)}$. Next, we have 
\begin{align*}
\text{deg }H'^{(i)}-2g(E_i)+2& \geq  \text{deg }H^{(i)}-k_i -2g(E_i)+2\\
& > \left( \frac{\ell+1}{s}-2\right) [E_i:\mathbb F_q(w)]-k_i.
\end{align*}
Once again, we have two possibilities for $k_i$, depending on whether $F=F_0$ or $F=E_v$ for some $v<i$. In the first case, $k_i=[E_i:F_0]$, and we obtain  
\begin{align*}
\text{deg }H'^{(i)}-2g(E_i)+2& >  \left( \frac{\ell+1}{s}-2\right) \frac{[E_i:E_0]}{\ell-1}-\frac{[E_i:E_0]}{\ell(\ell-1)}\\
& =\left( \frac{\ell+1-2s}{s(\ell-1)}-\frac{1}{\ell(\ell-1)}\right) [E_i:E_0],
\end{align*}
which is greater than zero using $\ell\geq 2s$. If $F=E_v$, we have that $k_i<[E_i:F_0]$, so that the l.h.s is strictly greater than zero. In both cases, we can use Proposition \ref{proposition classical code parameters} and obtain that $C_0$ has distance lower bounded by $\text{deg }H'^{(i)} -2g(E_i) = \text{deg }H^{(i)}-2g(E_i)-k_i = D_i^\perp-k_i$, where $D_i^\perp$ is the distance of $(C_{a,b}^{(i)})^\perp$.\par
Overall, we obtain that the distance of the quantum code is 
\[
d_i\geq \min(D_i,D_i^\perp)-k_i.
\] Moreover, under the assumption $\ell\geq 2s \geq 6$, we have that the lower bound for $D_i$ is greater than that of $D_i^\perp$, as can be seen in Proposition \ref{proposition classical code parameters}, so that $d_i\geq \text{deg }H^{(i)}-2g(E_i)-k_i$. To assess the asymptotic performance of the quantum code, we now compute a lower bound on the relative distance using the lower bound on the distance given in Proposition \ref{proposition classical code parameters}:
\begin{align*}
\frac{d_i}{n_i}\geq \frac{d_i}{N_i}\geq
& \left( \frac{\ell+1}{s(\ell-1)}\right)-\frac{k_i}{N_i} ,
\end{align*}
where we shall replace the value of $k_i$ depending on the case $F=F_0$ or $F=E_v$. Under the condition on $\ell$ and $s$, this is again constant, as for the code rate, which concludes the claim that the quantum code is asymptotically good.
\end{proof}

\subsection{Proof of Theorem \ref{theorem: main circuit} }\label{Appendix proof of main theorem}

\begin{proof}[Proof of Theorem \ref{theorem: main circuit}]
To prove the theorem, it is enough to study the action of the sequence of gates of the r.h.s of Equation \eqref{Equation physical circuit} on a product state of the form $\ket{\psi^0} \ket{\psi^1} \dots \ket{\psi^{m-1}}$, where $ \ket{\psi^j}\in \mathcal{Q}^{(i)}$ belongs to the $j$-th code block of $(\mathcal{Q}^{(i)})^{\otimes m}$, where we recall that $\mathcal{Q}^{(i)}=\text{CSS}(C_0,C_1)$, as defined in Equation~\eqref{equation quantum code def}. In turn, by linearity this is enough to check the behavior on a tensor product of computational basis states $\overline{\ket{x^0}}\:\overline{\ket{x^1}} \dots \overline{\ket{x^{m-1}}}$, where $x^j \in \mathbb{F}_q^{k_i}$. Recall that, given $x^j=(x^j_Q)_{Q\mid \zeta} \in \mathbb{F}_q^{k_i}$, the associated logical state can be expressed in the form $\overline{\ket{x^j}}\coloneqq\sum_{g \in \text{rowspan}(G_0)}\ket{\sum_{Q|\zeta} x^j_{Q}g_Q+g}$.  Therefore, by linearity again, it is enough to check the action of the circuit on a tensor product $$\ket{\mathsf F}\coloneqq\ket{f^0}\ket{f^1}\dots \ket{f^{m-1}},$$
where $\ket{f^j}$ is any summand $\ket{\sum_{Q|\zeta} x^j_{Q}g_Q+g}$, such that $g\in \text{rowspan}(G_0)=C_0^\perp$. It follows that $f^j$ is any codeword of $C_1$.\par
We may write any codeword $f^j\in C_1$ as a vector $f^j=(f^j_P)_{P\nmid\zeta}\in \mathbb F_q^{n_i}$, since the physical qubits of $\mathcal{Q}^{(i)}$ are in one to one correspondence with the places $P\nmid \zeta$. We may also write such a vector as a linear combination of the rows of $G$, 
\begin{align*}
f^{j} &\coloneqq \sum_{g\in \text{rows}(G)}x^j_g g\\
&=\sum_{Q \mid \zeta}x^j_Q g_Q+\sum_{g\in \text{rows}(G_0)}x^j_g g .
\end{align*}
where in the second equality, we have used that the rows of $G_1$ are in one-to-one correspondence with the logical qubits, them being also in one-to-one correspondence with the rational places $Q\mid \zeta$. To this vector corresponds a unique linear combination of rows of $\tilde{G}$, 
\begin{align*}
    \tilde{f}^j=& \sum_{g\in \text{rows}(\tilde G)} x^j_{\tilde g}\tilde{g}\\
    =&\sum_{Q \mid \zeta}x^j_Q \tilde g_Q+\sum_{\tilde g\in \text{rows}(G_0)}x^j_{\tilde g} \tilde g.
\end{align*}
We may identify each $\tilde f^j$ with a unique function, that we also denote $\tilde f^j$ , in the Riemann-Roch space $\mathcal{L}(H^{(i)})$. For any place $P\nmid \zeta$ we simply have that $$f^j_P= \tilde f^j(P)$$ and given any place $Q'\mid \zeta$ we have
\begin{align*}
    \tilde f^j(Q')=\sum_{Q\mid \zeta} x^j_{Q}\tilde g_Q(Q')=x^j_{Q'}.
\end{align*}
Now, applying the circuit in the r.h.s of Equation \eqref{Equation physical circuit} to $\ket{\mathsf F}$ yields
    \begin{align*}
\prod_{P\nmid \zeta}& \mathsf{C}^{m-1}\mathsf{Z}^{-\underline{u}_{P}   M (P)}[P,\sigma_1(P),\sigma_2(P),\dots,\sigma_{m-1}(P)]\:  \ket{\mathsf F}\\
        =&   \prod_{P\nmid \zeta}\exp\left( - \frac{2i\pi}{p}  \text{tr}\left( \underline{u}_{P}   M (P)   f^0_P    f^1_{\sigma_1(P)} \dots f^{m-1}_{\sigma_{m-1}(P)} \right) \right) \ket{\mathsf F} \\
        =&   \prod_{P\nmid \zeta}\exp\left(-  \frac{2i\pi}{p}  \text{tr}\left( \underline{u}_{P}   M (P)   \tilde{f}^0(P)    \tilde{f}^1(\sigma_1(P)) \dots \tilde{f}^{m-1}(\sigma_{m-1}(P)) \right) \right) \ket{\mathsf F} \\
        =& \exp \left(  -\frac{2i\pi}{p}  \text{tr} \left( \;  \sum_{P\nmid \zeta } \underline{u}_{P}   M (P)    \tilde{f}^0(P)    \tilde{f}^1(\sigma_1(P)) \dots \tilde{f}^{m-1}(\sigma_{m-1}(P)) \right) \right)  \ket{\mathsf F},
    \end{align*}
where we recall that $\mathrm{tr}\coloneqq\mathrm{tr}_{\mathbb F_q/\mathbb F_p}$ is the trace map from $\mathbb F_q$ to the prime field $\mathbb F_p$. Since the divisors $A^{(i)}$ and $B^{(i)}$ are invariant under any $\text{Gal}(E_i/F_0)$, the divisor $H^{(i)}=aA^{(i)}+bB^{(i)}$ is also invariant under any $\text{Gal}(E_i/F)\subseteq \text{Gal}(E_i/F_0)$. Moreover, note that $\tilde{f}(\sigma(P))$, which is the function $\tilde{f}$ evaluated at the place $\sigma(P)$, is the same thing as the function $\sigma^{-1}(\tilde{f})$ evaluated at the place $P$, which still belongs to the Riemann-Roch space $ \mathcal{L}(H^{(i)})$. Therefore, given any $\sigma \in\text{Gal}(E_i/F)$ and $\tilde f \in \mathcal{L}(H^{(i)})$, we have that $\sigma^{-1}(\tilde f) \in \mathcal{L}(H^{(i)})$. It follows that $\sigma_j^{-1}(\tilde{f}^j) \in \mathcal{L}(H^{(i)})$, for all $1\leq j\leq m-1$.\par
Crucially, recall that $M  \in \mathcal{L}(H^{(i)})$, so that Corollary \ref{Corollary multiplication property} applies to the $m$ functions $\sigma_j^{-1}(\tilde{f}^j)$ for $0\leq j\leq m-1$, with the function $\tilde f^m$ taken to be $M $. We obtain, for the sum inside the exponential,
    \begin{align*}
   \sum_{P\nmid \zeta }& \underline{u}_{P}   M (P)    \tilde{f}^0(P)    \tilde{f}^1(\sigma_1(P)) \dots \tilde{f}^{m-1}(\sigma_{m-1}(P)) \\
        =&\sum_{P\nmid \zeta } \underline{u}_{P}   M (P)     \tilde{f}^0(P)    \left(\sigma_1^{-1}(\tilde{f}^1)\right)(P) \dots \left(\sigma_{m-1}^{-1}(\tilde{f}^{m-1}) \right)(P) \\
        =&\sum_{Q\mid \zeta }\underline{u}_{Q}\;M (Q)    \tilde{f}^0(Q)    \left(\sigma_1^{-1}(\tilde{f}^1)\right)(Q) \dots \left(\sigma_{m-1}^{-1}(\tilde{f}^{m-1})\right)(Q) \\
        =&- \sum_{Q\in S}\gamma_Q   \tilde{f}^0(Q)   \tilde{f}^1(\sigma_1(Q))\dots\tilde{f}^{m-1}(\sigma_{m-1}(Q))\\
        =& -\sum_{Q\in S}\gamma_Q   x^0_{Q}   x^1_{\sigma_1(Q)}\dots x^{m-1}_{\sigma_{m-1}(Q)},
    \end{align*}
where in the fourth line, we applied Corollary \ref{Corollary multiplication property}, and in the last line have used that $$M (Q)=\sum_{Q'\in S}\gamma_{Q'}   (\underline{u}_{Q'})^{-1}  \tilde g_{Q'}(Q)=\sum_{Q'\in S}\gamma_{Q'}   (\underline{u}_{Q'})^{-1}  \delta_{Q',Q}.$$
It follows directly that
 \begin{align}\label{Equation effect of circuit on arbitrary vectors}
       \prod_{P\nmid \zeta}\mathsf{C}^{m-1}&\mathsf{Z}^{-\underline{u}_{P}   M (P)}[P,\sigma_1(P),\sigma_2(P),\dots,\sigma_{m-1}(P)] \:\ket{\mathsf F}\\
    &   =\prod_{Q\in S}\exp \left(  \frac{2i\pi}{p} \text{tr}\left( \gamma_Q   x^0_{Q}   x^1_{\sigma_1(Q)}\dots x^{m-1}_{\sigma_{m-1}(Q)}\right)\right)\ket{\mathsf F}\\
\end{align}
This holds for any functions $f^0,f^1,\dots,f^{m-1}$, and hence it holds for any product $\overline{\ket{\boldsymbol{x}}}\coloneqq\overline{\ket{x^0}}\:\overline{\ket{x^1}}\dots\overline{\ket{x^{m-1}}}$ of logical basis states and we obtain

\begin{align}
       \prod_{P\nmid \zeta}\mathsf{C}^{m-1}&\mathsf{Z}^{-\underline{u}_{P}   M (P)}[P,\sigma_1(P),\sigma_2(P),\dots,\sigma_{m-1}(P)] \overline{\ket{\boldsymbol{x}}}\\
    &   =\prod_{Q\in S}\exp \left(  \frac{2i\pi}{p} \text{tr}\left( \gamma_Q   x^0_{Q}   x^1_{\sigma_1(Q)}\dots x^{m-1}_{\sigma_{m-1}(Q)}\right)\right)\overline{\ket{\boldsymbol{x}}}\\
    &= \prod_{Q\in S}\overline{\mathsf{C}^{m-1}\mathsf{Z}^{\gamma_Q}[Q,\sigma_1(Q),\sigma_2(Q),\dots,\sigma_{m-1}(Q)]}\:\overline{\ket{\boldsymbol{x}}}.
\end{align}
Since this holds for products of logical basis states, it also holds for superpositions thereof.
\end{proof}

\bibliographystyle{alpha}
\newcommand{\etalchar}[1]{$^{#1}$}

\end{document}